%% file: main.tex
\documentclass[11pt]{article}
\usepackage[hmargin=1in,vmargin=1in]{geometry}
\usepackage{hchang}
\usepackage{thm-restate}
\usepackage{cleveref}
\usepackage{comment}

\newtheorem{theorem}{Theorem}[section]
\newtheorem{lemma}[theorem]{Lemma}

\newtheorem{definition}[theorem]{Definition}
\newtheorem{question}[theorem]{Question}

\def\eps{\epsilon}

\newcommand{\ball}{\mathbf{B}}

\begin{document}
	
\title{Three trees suffice for a constant stretch in minor-free graphs}

\author{%
    Hung Le%
\thanks{Manning CICS, UMass Amherst. Email: {\tt hungle@cs.umass.edu}}  
\and
Huy Pham
\thanks{Hanoi University of Science and Technology. Email: {\tt huy.phamquang@hust.edu.vn}}
\and
Cuong Than%
\thanks{Manning CICS, UMass Amherst. Email: {\tt cthan@umass.edu}}  
\and
Tuan Tran
\thanks{School of Mathematical Sciences, University of Science and Technology of China. Email: {\tt trantuan@ustc.edu.cn}}
}

\date{}

\maketitle
	
\thispagestyle{empty}

\begin{abstract}
    In this short note, we show that $H$-minor-free graphs have a tree cover with $3$ trees and constant stretch for any fixed graph $H$. The number of trees matches the recent lower bound by  Chen, Tan, and Xu~\cite{CTX26} who showed that a toroidal grid requires at least $3$ trees for constant stretch. Our result is obtained by establishing a connection between tree covers and Assouad--Nagata dimension and then invoking the recent dimension bound for minor-free metrics by Liu~\cite{Liu25}.
\end{abstract}

\input{intro}
\input{prem}

\input{3treecover}

\section*{AI Disclosure}

The authors only use AI tools (ChatGPT/Gemini) for proofreading the papers. The write-up and ideas are all by the human authors. Perhaps results like this short note are no longer difficult to obtain with the recent advanced AI tools. Yet it took us a few months to notice the connection established here, which we are very excited about.

\section*{Acknowledgment}

This work was initiated at the Application Driven Mathematics (ADM) 2025 summer program organized by the VinBigdata Research Institute. Hung Le and Cuong Than are supported by the NSF CAREER award CCF-2237288 and the NSF grants CCF-2517033. Cuong Than is also supported by a Google Ph.D. Fellowship.
Tuan Tran is supported by the Excellent
Young Talents Program (Overseas) of the National Natural Science Foundation of China under Grant No.
GG0010007003.
\bibliographystyle{plain}
\bibliography{refs, FOCS23}
\end{document}

%% file: intro.tex
\section{Introduction}

Let $(X,\delta)$ be a metric space. A tree cover of $X$ is a collection $\mathcal{T}$ of edge-weighted trees such that for every tree $T\in\mathcal{T}$, $(T,d_T)$ \EMPH{dominates} $(X,\delta)$; that is, $d_T(u,v)\geq\delta(u,v)$ for all $u,v\in X$. The cover $\mathcal{T}$ has \EMPH{stretch $t$} if, for every pair $u,v\in X$, there exists a tree $T\in\mathcal{T}$ such that $d_T(u,v)\leq t\cdot\delta(u,v)$. The same definition applies to graphs equipped with their shortest-path metrics. 

In algorithm design, tree covers allow one to reduce solving an algorithmic problem on (complicated) metric spaces and graphs to solving the problem on trees, which typically admits simple and efficient solutions. Examples abound, including distance oracles, labeling and routing schemes, spanners, and matching, among others~\cite{TZ05,KLMS22,ACRX22,CCL+23a,CCL+23b,CCL+26,LMS+26}. Thus, tree covers have been studied for decades~\cite{AP92,AKPW95,ADM+95,GKR01,MN06,DYL06,CGMZ16,BFN22,CCL+23a,CCL+23b,CCL+26,LLS+26,CTX26,Hua26,LMS+26,BKT26}.

\paragraph{The fixed stretch paradigm.} In algorithmic applications of tree covers, the stretch often controls the \EMPH{quality} of the solution, while the number of trees dictates the \EMPH{cost} of obtaining the solution, which is either time, space, or both. For example,
a distance oracle constructed from a tree cover with $k$ trees and stretch $t$ has space $O(kn)$, query time $O(k)$, and distance approximation factor $t$. From this point of view, unsurprisingly, the vast majority of the literature has focused on the \EMPH{fixed-stretch paradigm}: fixing a stretch factor $t$, and constructing a tree cover with a small number of trees (as a function of $t$ and $n$). This question has been extensively investigated and well understood for various metric spaces and graph classes~\cite{TZ05,MN06,ADDJS93,ADM+95,CGMZ16,LLS+26,GKR01,BFN22,CCL+23a,CCL+23b}. A sample of representative results includes a tree cover for general metrics with stretch $k$ and $O(kn^{O(1/k)})$ trees~\cite{TZ05,MN06}, which can be made Ramsey~\cite{MN06}, a tree cover for Euclidean/doubling metrics with stretch $1+\epsilon$ and $\epsilon^{-O(d)}$ trees for any fixed $\epsilon\in (0,1)$~\cite{ADM+95,CGMZ16,BFN22,CCL+26,LLS+26}, and a tree cover for planar/minor-free metrics with stretch $1+\epsilon$ and $\eps^{-O(1)}$ trees for any fixed $\eps\in (0,1)$~\cite{CCL+23a,CCL+23b}. These results are either tight or nearly tight.  

An orthogonal question is to fix the number of trees to be $k\geq 2$, and construct a tree cover with small stretch (as a function of $k$). Surprisingly, this question has only recently been explored and remains wide open. It falls into the \EMPH{fixed cost paradigm} where fixing the number of trees is equivalent to fixing the cost (time or space) of downstream algorithmic tasks.  Our result here is also situated in this paradigm.  

\paragraph{The fixed cost paradigm.} For general metrics, the pioneering work of Bartal, Fandina, and Neiman constructed a (Ramsey) tree cover of size $k$ and stretch $O(n^{1/k}\log^{1-1/k}(n))$ for any $k\geq 1$. For the lower bound, Chen, Tan, and Xu~\cite{CTX26} made a surprising connection to combinatorial topology and showed a lower bound $\Omega(n^{1/2^{k-1}})$ on the stretch for any constant $k\geq 1$. Their lower bounds are almost tight for $k = 2$ (the case $k = 1$ is trivial). The lower bound has been recently improved to  $\Omega(n^{1/O(k^2)})$~\cite{Hua26}. However, (nearly) tight bound, which likely is $\Omega(n^{1/k})$, remains wide open for any $k\geq 3$. 

For well-studied structured classes of metrics and graphs, such as Euclidean/doubling metrics and planar/minor-free graphs, as the number of trees increases from $1$ to $+\infty$, we expect the stretch reduces to $O(1)$ at some point (the stretch constant could depend on the structural parameter of the input metric/graph such as the dimension or the size of the excluded minor).  Thus, a basic question in these cases is: 

\begin{question}\label{quest:basic}
What is the minimum number of trees needed to obtain a constant stretch for Euclidean /doubling metrics (of constant dimension) and planar/minor-free graphs? 
\end{question}

Of course, knowing how the stretch degrades as a function of the number of trees would also answer \Cref{quest:basic}, but we are very far from that goal.

For point sets in  $\mathbb{R}^d$ for a constant $d$, the answer to \Cref{quest:basic} seems to be exactly $d$ trees, but proving it remains challenging. A pair of recent papers~\cite{LMS+26,BKT26} gave the answer for the $d=2$ case, showing that $2$ trees are sufficient (and necessary) for a constant stretch. For $d = 3$, by embedding the toroidal grid into $\mathbb{R}^3$ with constant distortion, the lower bound for the toroidal grid by Chen, Tan, and Xu~\cite{CTX26} implies that 3 trees are necessary for constant stretch (while the current upper bound is $5$ trees~\cite{Cha98}). For constant $d\geq 3$,   Chan~\cite{Cha98} constructed a tree cover with $2\ceil{d/2}+1 \approx d+1$ trees and constant stretch, but the lower bound is very far from $d$~\cite{CTX26,BKT26}. An interesting, but perhaps not surprising, takeaway is that \EMPH{the number of trees needed for constant stretch increases as the dimension increases}. 

For minor-free graphs, \Cref{quest:basic} is wide open. Does the number of trees needed increase with the size of the minor, similar to the Euclidean case?  Chen, Tan, and Xu~\cite{CTX26} showed that for the torus grid, at least $3$ trees are needed for constant stretch. There is no matching upper bound, even for the toroidal grid graph. Since the torus has genus $g = 1$, it is conceivable that one could show a stronger lower bound for grid graphs embedded in a surface with genus $g \geq 2$. In particular, the number of trees could increase as a function of the genus. Such a result would imply that the number of trees needed increases with the size of the minor for minor-free graphs. Perhaps surprisingly, in this short note, we show that the number of trees \EMPH{stays the same as the size of the excluded minor increases}. In particular, \EMPH{3 trees suffice} for any graph class excluding a fixed minor, matching the lower bound by Chen, Tan, and Xu~\cite{CTX26} for toroidal graphs.

\begin{theorem}\label{thm:main}
For every fixed graph $H$, there exists a constant $t_H$ such that every $H$-minor-free graph admits a tree cover of size $3$ and stretch at most $t_H$. Furthermore, for every integer $w > 0$, there exists a constant $t_w$ such that every graph with treewidth bounded by $w$ admits a tree cover of size $2$ and stretch $t_w$.
\end{theorem}

Our \Cref{thm:main} completely resolves \Cref{quest:basic} for graphs excluding a $K_h$ as a minor for any constant $h\geq 8$. (The toroidal graphs exclude $K_8$ as a minor.)

We prove \Cref{thm:main} by making a new connection between Assouad--Nagata dimension and a small tree cover. In more technical terms, a standard technique for constructing a tree cover is to construct a hierarchical partition family (HPF) with a small number of hierarchical partitions~\cite{BFN22,CCL+23a, CCL+23b, CCL+25}. Roughly speaking, an HPF is a small collection of hierarchical partitions such that every sufficiently small ball is contained in a cluster of at least one partition at the appropriate scale.  If the Assouad--Nagata dimension is $k$, we can transform a $k$-dimensional control function into what we call \EMPH{colorable sparse covers}.  Then we can combine these covers across different scales to construct an HPF of size $k + 1$, yielding a tree cover of the same size.  By establishing this connection between the Assouad--Nagata control function and tree covers, we can now invoke recent results by Liu~\cite{Liu25}, who showed that proper minor-closed classes of graph metrics have Assouad--Nagata dimension at most $2$, whereas classes of bounded-treewidth graphs have Assouad--Nagata dimension at most $1$ \cite{Liu25}, giving \Cref{thm:main}.

%% file: prem.tex
\section{Preliminaries}

We use the standard notation $[n]$ for the set $\{1, 2, 3, \ldots, n\}$.
Let $(X, \delta)$ be a finite metric space. For every $w \in X$ and every $r > 0$, $\ball_X[w, r]$ denotes the set of points in $X$ whose distance from $w$ is at most $r$. When $X$ is clear from the context, we omit the subscript and write $\ball[w, r]$.
For any two subsets $Y$ and $Z$ of $X$, define $\delta(w, Y) = \delta(Y, w) = \min_{u \in Y}\delta(w, u)$ and $\delta(Y, Z) = \min_{u \in Y}\delta(u, Z)$. Two points $u,v\in Y$ are said to be \EMPH{$r$-connected} in $Y$ if there exists a path from $u$ to $v$ in the weighted complete graph $\left(Y,\binom{Y}{2},\delta\right)$ such that every edge on the path has weight at most $r$. Being $r$-connected is an equivalence relation on $Y$ and therefore partitions $Y$ into equivalence classes. Each such equivalence class is called an \EMPH{$r$-component} of $Y$.

There are several equivalent formulations of the Assouad--Nagata dimension. We use the $r$-component formulation of a $k$-dimensional control function; see \cite[Definitions~2.1 and~3.8]{BDLM08} or \cite[Section~4.2]{BBE+23}.

\begin{definition}[Assouad--Nagata Dimension]
    Let $(X, \delta)$ be a metric space. A \EMPH{$k$-dimensional control function} of $(X, \delta)$ is a function $f: \mathbb{R}^+ \rightarrow \mathbb{R}^+$ such that, for every $r > 0$, there exist subsets $U_1, U_2, \ldots, U_{k + 1}$ satisfying the following properties.
\begin{itemize}
\item \textbf{Covering.} $\bigcup_{i = 1}^{k + 1} U_i = X$.
\item \textbf{Bounded Diameter.} For every $i \in [k + 1]$, each $r$-component $U$ of $U_i$ has diameter at most $f(r)$; i.e., for every $u, v \in U$, $\delta(u, v) \leq f(r)$.
\end{itemize}
The function $f$ is a \EMPH{dilation} if there exists a constant $c$ such that $f(r) \leq cr$ for every $r \in \mathbb{R}^+$. The Assouad--Nagata dimension of $(X, \delta)$ is the smallest integer $k$ such that $(X, \delta)$ admits a $k$-dimensional control function that is also a dilation.
\end{definition}

A graph class $\mathcal{G}$ is said to have Assouad--Nagata dimension at most $k$ if there is an universal constant $c_\mathcal{G}$ such that, for any $G \in \mathcal{G}$, the distance metric of $G$ admits a $k$-dimensional control function $f_\mathcal{G}(r) = c_\mathcal{G}r$. Note that $c_\mathcal{G}$ does not depend on any member $G$ of $\mathcal{G}$. In partitcular, any $H$-minor-free graph has a $2$-dimensional control function $f_H(r) = c_Hr$ for some constant $c_H$ depending on $H$. Throughout this note, all graphs are finite and connected. 

In this paper, we study the relationship between the Assouad--Nagata dimension and the existence of tree covers of small size.

\begin{definition}[Tree Cover]
    Let $(X, \delta)$ be a metric space. A \EMPH{$\sigma$-tree cover} of $(X, \delta)$ is a collection of edge-weighted trees $\mathcal{T} = \{T_1, T_2, \dots, T_l\}$, each of which has vertex set containing $X$, satisfying the following properties:

\begin{enumerate}
    \item For every tree $T_i \in \mathcal{T}$, its distance dominates the original metric:
    \begin{equation*}
        d_{T_i}(u, v) \geq \delta(u, v) \quad \forall u, v \in X.
    \end{equation*}

    \item For every pair of points $u, v \in X$, there exists a tree $T_j \in \mathcal{T}$ in which their distance is approximated within a factor of $\sigma$:
    \begin{equation*}
        d_{T_j}(u, v) \leq \sigma \cdot \delta(u, v).
    \end{equation*}
\end{enumerate}
The values $\sigma$ and $l$ are called the \EMPH{stretch} and \EMPH{size} of $\mathcal{T}$, respectively.
\end{definition}

A tree cover is typically constructed using a family of hierarchical partitions. In this work, we use the following notion of a hierarchical partition family (HPF), introduced by Bartal, Fandina, and Neiman \cite{BFN22}. 

\begin{definition}[$(\mu, \rho)$-Hierarchical Partition Family]
    Let $(X, \delta)$ be a finite metric in which the minimum distance between any two distinct points is $2$. A \EMPH{$\mu$-hierarchical partition}, or \EMPH{$\mu$-HP}, is a hierarchical partition of $X$, denoted by $\mathcal{P} = \{\mathcal{P}^0, \mathcal{P}^1, \ldots, \mathcal{P}^h\}$, satisfying the following properties:

\begin{enumerate}
    \item $\mathcal{P}^0 = \{\{u\} : u \in X\}$ and $\mathcal{P}^h = \{X\}$.
    \item Each $\mathcal{P}^i$ is a partition of $X$ into subsets, called clusters, each of which has diameter at most $\mu^i$.
    \item For $i < h$, $\mathcal{P}^i$ is a refinement of $\mathcal{P}^{i + 1}$; that is, each cluster in $\mathcal{P}^{i + 1}$ is a union of clusters in $\mathcal{P}^{i}$. 
\end{enumerate}

A \EMPH{$(\mu, \rho)$-hierarchical partition family}, or \EMPH{$(\mu, \rho)$-HPF}, denoted by $\mathfrak{P} = \{\mathcal{P}_1, \mathcal{P}_2, \ldots, \mathcal{P}_l\}$, is a collection of $\mu$-HPs such that, for every point $u \in X$ and every nonnegative integer $j \leq h$, there exists an index $i \in [l]$ for which some cluster in $\mathcal{P}_{i}^j$ contains $\ball_X[u, \mu^j/\rho]$. The number $l$ is called the size of the HPF. 
\end{definition}

The connection between HPF and tree covers is showed in \cite{BFN22}.

\begin{lemma}[Lemma 8 in~\cite{BFN22}]
    \label{lm:hpf-treecover}
    Given a metric space $(X, \delta)$, if $X$ admits a $(\mu, \rho)$-HPF of size $k$, then $X$ has a tree cover with stretch $\mu\rho$ and size $k$. 
\end{lemma}

%% file: 3treecover.tex
\section{Tree Cover Construction}

To show that $(X, \delta)$ has a tree cover of size $k + 1$, we first show that $(X, \delta)$ admits an HPF of size $k + 1$, and then construct a tree cover from the HPF. Throughout this section, we assume that $|X| > 1$ and rescale the metric so that the minimum distance between any two distinct points of $X$ is $2$. Our main theorem is:

\begin{theorem}
    \label{thm:andim-treecover}
    Let $(X, \delta)$ be a finite metric space that admits a $k$-dimensional control function $f(x) \leq c \cdot x$ for some positive constant $c$, then $X$ has a tree cover of size $k + 1$ and stretch $12(3c + 2)^2$. 
\end{theorem}

We show how to construct a $(\mu, \rho)$-HPF for a metric space of Assouad--Nagata dimension $k$.
Equivalently we show that, if a metric space $(X, \delta)$ admits a $k$-dimensional control function $f$ such that $f$ is a dilation, then $(X, \delta)$ also admits an HPF of size $k + 1$. To do that, we construct sparse covers of $X$ for multiple distance scales, then use sparse covers to build an HPF. 

\begin{definition}[Sparse Cover]
    Given a metric space $(X, \delta)$, a \EMPH{$(\rho, k, \Delta)$-sparse cover} of $X$ is a set of clusters $\mathcal{A} = \{C_1, C_2, \ldots C_l\}$, each is a subset of $X$, satisfying the following properties:
    \begin{enumerate}
        \item \textbf{Padding}: for every $u \in X$, there is an $C_i$ containing every point in $\ball_X[u, \Delta/\rho]$, and
        \item \textbf{Sparsity}: for every $u \in X$, there are at most $k$ clusters containing $u$, and
        \item \textbf{Bounded Diameter}: for every $i$, each $C_i$ has diameter at most $\Delta$.
    \end{enumerate}
    The number $k$ is the sparsity of the sparse cover. 
\end{definition}

A $(\rho, k, \Delta)$-sparse cover $\mathcal{A}$ of $X$ is \EMPH{colorable} if we can partition $\mathcal{A}$ into $k$ subsets $\{\mathcal{A}_1, \mathcal{A}_2, \ldots \mathcal{A}_k\}$, such that the distance between any two sets in $\mathcal{A}_i$ is greater than $\frac{\Delta}{\rho}$. 

\begin{lemma}[Assouad--Nagata Dimension and Colorable Sparse Cover]
    \label{lm:controlfunc_sparsecover}
    Let $(X, \delta)$ be a finite metric space that admits a $k$-dimensional control function $f(x) \leq c\cdot x$, for some positive constant $c$, then for every $r > 0$, $X$ admits a colorable $(3c + 2, k + 1, (3c + 2)r)$-sparse cover.
\end{lemma}

\begin{proof}
    By definition of $k$-dimensional control function, there exists a set $\mathcal{U} = \{U_1, U_2, \ldots U_{k + 1}\}$ such that, each $3r$-component of every $U_i$ has diameter bounded by $f(3r)$. For every $i$, let $L(U_i)$ be the set of points in $\bigcup_{u \in U_i}\ball_X[u, r]$. Let $\mathcal{A}_i$ be the set of clusters, each is an $r$-component in $L(U_i)$ and $\mathcal{A} = \bigcup_{i \in [k + 1]}\mathcal{A}_i$. We show that $\mathcal{A}$ is a $(3c + 2, k + 1, (3c + 2)r)$-sparse cover. Furthermore, we show that $\mathcal{A}$ is colorable by proving the distance between any two sets in $\mathcal{A}_i$ is greater than $r$ for every $i$. 
    \paragraph{Sparsity and Colorability.} By the definition of an $r$-component, the distance between any two clusters in $\mathcal{A}_i$ is greater than $r$. Moreover, the clusters in $\mathcal{A}_i$ are pairwise disjoint, so every point of $X$ belongs to at most one cluster in $\mathcal{A}_i$. Since there are $k + 1$ collections $\mathcal{A}_1, \ldots, \mathcal{A}_{k + 1}$, every point belongs to at most $k + 1$ clusters in $\mathcal{A}$.
   \paragraph{Padding.}For every $u \in X$, there exists an $i \in [k + 1]$ such that $u \in U_i$. Therefore, $\ball_X[u, r] \subseteq L(U_i)$.  Moreover, every point of $\ball_X[u, r]$ is at distance at most $r$ from
    $u$. Therefore, $\ball_X[u, r]$ is contained in a single $r$-component of
    $L(U_i)$. Hence, there exists a cluster in $\mathcal{A}_i$ containing
    $\ball_X[u, r]$. 
   
   \paragraph{Bounded Diameter.}Fix $i \in [k + 1]$. Consider an $r$-component $C$ of $L(U_i)$. We show that $\diam(C) \leq f(3r) + 2r$. We show by contradiction that $C$ cannot have non-empty intersection with more than one $3r$-components of $U_i$. Assume otherwise, let $C'_1$ be a $3r$-component that intersects $C$. Such a component must exist: since $C$ is nonempty, the definition of $L(U_i)$ implies that there exists a point $u \in U_i$ such that $C \cap \ball_X[u, r] \neq \emptyset$. Hence, $C$ must contain $u$ by definition of an $r$-component. We abuse the notation and let $L(C'_1) = \bigcup_{u \in C'_1}\ball_X[u, r]$.
   
   We prove that $C \subseteq L(C'_1)$. Let $w$ be a point in $C \cap L(C'_1)$ and $w'$ be a point in $C \setminus L(C'_1)$. By definition of $r$-component, there is a path $P$ from $w$ to $w'$ in $\left(X, \binom{X}{2}, \delta\right)$ such that $P$ contains only points in $C$ and any edge in $P$ has weight at most $r$. Let $(u, v)$ be the first edge in $P$ such that $u \in L(C'_1)$ and $v \not\in L(C'_1)$. The edge $(u, v)$ exists since $w \in L(C'_1)$ and $w' \not\in L(C'_1)$. 
   Hence, $\delta(u, v) \leq r$. 
   Let $u'$ be the point in $C'_1$ such that $u \in \ball_X[u', r]$. 
   Let $v'$ be the point in $U_i$ such that $v \in \ball_X[v', r]$. If there are multiple such $u', v'$, choose them arbitrarily.
   If $v' \in C'_1$, $v \in L(C'_1)$ by definition of $L(C'_1)$, contradicting the choice of $v\notin L(C'_1)$. Thus, $v'$ is in another $3r$-component of $U_i$. 
   Since $u'$ and $v'$ are in different $3r$-components, $\delta(u', v') > 3r$. By triangle inequality,
   \begin{equation*}
       \delta(u, v) \geq \delta(u', v') - \delta(u, u') - \delta(v, v') > 3r - r - r = r, 
   \end{equation*}
   a contradiction. Therefore, $C$ intersects only one $3r$-component $C_1'$ of $U_i$, implying $C \subseteq L(C_1')$. The diameter of $L(C_1')$ is at most $f(3r) + 2r$ by triangle inequality. Then, $\diam(C) \leq f(3r) + 2r \leq (3c + 2)r$.
\end{proof}

\begin{lemma}[Colorable Sparse Cover Implies HPF]
    \label{lm:sparsecove_hpf}
    Let $(X, \delta)$ be a finite metric space. Suppose that, for some constant $c \geq 2$ and every $r > 0$, $X$ admits a colorable $(c, k, cr)$-sparse cover. Then, $X$ admits a $(6c, 2c)$-HPF of size $k$. 
\end{lemma}

\begin{proof}
    Recall that minimum distance between any two distinct points in $X$ is $2$. Let $\mu = 6c$. We construct an HPF as follows: 
    \begin{enumerate}
        \item Let $h = \max\{1, \lceil\log_{\mu}\diam(X)\rceil\}$. For each $i$ from $0$ to $h$, let $r_i = \mu^i/2c$. For $i \in [h - 1]$, construct a colorable $(c, k, cr_i)$-sparse cover, denoted by $\mathcal{A}^i$. By colorability, $\mathcal{A}^i$ can be partitioned into subcollections $\mathcal{A}^i_1, \mathcal{A}^i_2, \ldots, \mathcal{A}_{k}^i$ such that the clusters in each subcollection are at distance greater than $r_i$ from one another. 
        \item For each $j \in [k]$, construct an HP $\mathcal{P}_j = \{\mathcal{P}^0_j, \mathcal{P}^1_j, \ldots, \mathcal{P}^h_j\}$ from $\mathcal{A}^1_j, \mathcal{A}^2_j, \ldots, \mathcal{A}^{h - 1}_j$ as follows. First, let $\mathcal{P}^0_j = \{\{u\}: u \in X\}$ and $\mathcal{P}^h_j = \{X\}$. For $i \in [h - 1]$, each cluster in $\mathcal{P}_j^i$ is either the union of clusters in $\mathcal{P}^{i - 1}_j$ that intersect the same cluster in $\mathcal{A}^i_j$ or a single cluster in $\mathcal{P}^{i - 1}_j$. Specifically, for each cluster $C$ in $\mathcal{P}^{i - 1}_j$, mark a cluster in $\mathcal{A}^i_j$ that intersects $C$. If no such cluster exists, add $C$ to $\mathcal{P}_j^{i}$. We later show by induction that each cluster $C$ in $\mathcal{P}^{i - 1}_j$ has diameter at most $\mu^{i - 1} = r_i/3$ and any two clusters in $\mathcal{A}^i_j$ are at distance greater than $r_i$ from each other. Then, $C$ can intersect at most one cluster in $\mathcal{A}^i_j$. Finally, for each cluster $F$ in $\mathcal{A}_j^i$, create a cluster in $\mathcal{P}^i_j$ by taking the union of all clusters in $\mathcal{P}^{i - 1}_j$ that mark $F$.  
    \end{enumerate}

    By induction on $i$, each $\mathcal{P}_j^i$ is a partition of $X$. Moreover, by construction, each cluster in $\mathcal{P}_j^i$ is a union of clusters in $\mathcal{P}_j^{i - 1}$. We next show by induction that every cluster in $\mathcal{P}_j^i$ has diameter at most $\mu^i$. The claim holds for $i = 0$, since $\mathcal{P}_j^0$ contains only singleton clusters. For $i \leq h - 1$, assume that every cluster in $\mathcal{P}_j^{i - 1}$ has diameter at most $\mu^{i - 1}$. Consider a cluster $C$ in $\mathcal{P}_j^{i}$. If $C$ also belongs to $\mathcal{P}_j^{i - 1}$, then its diameter is at most $\mu^{i - 1} < \mu^i$. Otherwise, $C$ is an union of clusters in $\mathcal{P}_j^{i - 1}$ that intersect some $F$ in $\mathcal{A}_j^i$. Note that $C$ contains $F$ since $\mathcal{P}_j^{i - 1}$ is a partition of $X$. In this case, the diameter of $C$ is at most $\diam(F) + 2\mu^{i - 1} \leq \mu^i/2 + 2\mu^{i - 1} \leq \mu^i$. For $i = h$, $\mathcal{P}_j^h$ consists of the single cluster $X$, whose diameter satisfies $\diam(X) = \mu^{\log_\mu(\diam(X))} \leq  \mu^h$. 

    Thus, each $\mathcal{P}_j$ is a $\mu$-HP of $X$. Let $\mathfrak{P} = \{\mathcal{P}_1, \mathcal{P}_2, \ldots, \mathcal{P}_k\}$. We now verify that $\mathfrak{P}$ satisfies the padding property. For every $u \in X$ and every nonnegative $i \leq h$, consider the ball $\ball_X[u, \mu^i/2c] = \ball_X[u, r_i]$. If $i = 0$, $\ball_X[u, r_i]$ contains only the point $u$ and is therefore contained in a cluster $\{u\} \in \mathcal{P}^0_j$ for every $j \in [k]$. If $i \in [h - 1]$, there exist an index $j \in [k]$ and a cluster $F \in \mathcal{A}^i_j$ containing $\ball_X[u, r_i]$. By construction, some cluster in $\mathcal{P}^i_j$ contains $F$ and hence also contains $\ball_X[u, r_i]$. For $i = h$, $\ball_X[u, r_i]$ is contained in the cluster $X$ in the highest-level partition $\mathcal{P}^h_j = \{X\}$ for any $j$. Therefore, $\mathfrak{P}$ is a $(6c, 2c)$-HPF of size $k$. 
\end{proof}

We now ready to prove our main theorem.

\begin{proof}[Proof of \Cref{thm:andim-treecover}]
    By \Cref{lm:controlfunc_sparsecover}, we obtain that for any distance $r > 0$, $X$ admit a colorable $(3c + 2, k + 1, (3c + 2)r)$-sparse cover. From \Cref{lm:sparsecove_hpf}, $X$ has a $(6(3c + 2), 2(3c + 2))$-HPF of size $k + 1$. Combining with \Cref{lm:hpf-treecover}, we conclude that $X$ has a tree cover of size $k + 1$ and stretch $12(3c + 2)^2$.
\end{proof}

\Cref{thm:main} follows directly from \Cref{thm:andim-treecover}.